\documentclass[letterpaper, 10 pt, conference]{ieeeconf}

\IEEEoverridecommandlockouts
\usepackage{times}
\usepackage{fancyhdr}
\usepackage{float}
\usepackage{cite}

\usepackage{listings}
\usepackage{booktabs}
\usepackage{subfig, bm}

\usepackage{tikz}
\usepackage{amssymb}
\usepackage{dsfont}
\usepackage{mdwlist}
\usepackage{hyperref}

\hypersetup{colorlinks,linkcolor={blue},citecolor={blue}} 
\usepackage{amsmath,graphicx,xcolor}
\usepackage{amsthm}
\usepackage{subfiles}
\usepackage{import}
\usepackage{dsfont}
\usepackage{xfrac}
\usepackage[font=small,labelfont=bf]{caption}
\usepackage{mathtools}
\usepackage{algorithm}
\usepackage{algpseudocode}
\usepackage{adjustbox}
\usepackage{pifont}
\usepackage{thmstyles}
\usepackage{stfloats}
\usepackage{cuted}
\usepackage{color}
\usepackage{multirow}
\usepackage{makecell} 

\algrenewcommand\algorithmicrequire{\textbf{Input:}}
\algrenewcommand\algorithmicensure{\textbf{Output:}}

\title{\LARGE \bf
Communication Architecture Co-Design for Distributed Control \\ via System Level Synthesis
}

\author{Jeonghyeon Noh and SooJean Han$^*$
\thanks{$^*$Corresponding author.Authors are with School of Electrical Engineering, Korea Advanced Institute of Science \& Technology (KAIST), Daejeon, Republic of Korea. E-mail: \texttt{\{jeonghyeon72, soojean\}@kaist.ac.kr}. The code of this paper is available at \url{https://github.com/outisKim/comm-arch-codesign-sls}.}
\thanks{This work was supported by the Institute of Information \& Communications Technology Planning \& Evaluation (IITP) grant funded by the Korea government (MSIT) (No. RS-2025-25442469, Development of Edge AI Server Technology with High-Performance and High-Reliability).}
}

\begin{document}
\maketitle


\thispagestyle{empty}
\pagestyle{empty}


\begin{abstract}
As networked dynamical systems grow in scale, distributed controller architecture design has become a critical issue. 
To address this, prior studies have incorporated the regularization for design (RFD) framework into system-level synthesis (SLS), focusing primarily on actuator and sensor placement. 
However, communication architectures also fundamentally impact the performance of distributed control, and scalable explicit co-design of the underlying communication topology within SLS remains challenging.
This paper proposes a scalable convex framework for the co-design of communication architectures and distributed controllers within the SLS paradigm. We explicitly characterize a communication topology-induced feasible subspace of system responses, parameterized by the communication adjacency matrix and separate from the physical plant topology.
By deriving a communication link norm from this characterized subspace, we formulate a convex SLS problem for communication architecture co-design.
A proximal ADMM scheme is employed to appropriately decompose the optimization problem in a manner that preserves scalability to large networks.
Numerical experiments validate the efficacy of the proposed framework in optimizing the trade-off between closed-loop performance and communication architecture complexity, while also confirming its scalability.
\end{abstract}


\section{INTRODUCTION}
Large-scale networked dynamical systems--such as power grids, transportation networks, and multiagent robotic systems--consist of many interacting subsystems whose physical couplings can both enable coordinated behavior and propagate disturbances. 
As a result, distributed feedback control is often essential for stability and performance, especially when centralized architectures are impractical due to sensing, communication, and computational constraints. 
A key challenge in such systems is the trade-off between closed-loop performance and communication requirements: richer information exchange can improve performance, but it also incurs costs in bandwidth, latency, reliability, and energy~\cite{ballotta2023can,ballotta2025role,shin2023near}. 
To this end, \emph{system level synthesis (SLS)} provides a natural framework for studying this trade-off by enabling scalable convex synthesis of distributed controllers directly in terms of closed-loop system responses subject to achievability and locality constraints~\cite{wang2014localized,anderson2019system}.

However, in many SLS-based designs, the communication structure is fixed a priori or aligned with the physical interconnection topology of the plant~\cite{khalil2025distributed,alonso2022distributed,anderson2019system,wang2016localized}.
Recent SLS-based work reduces communication usage under a given architecture \cite{yang2025safe,li2020separating}, but does not treat the communication topology itself as a design variable.
Modern cyber-physical systems, however, often support configurable or augmentable communication (e.g., software-defined routing, scheduled wireless links, or additional ``shortcut'' links),
raising the communication topology co-design problem: how should one \emph{jointly design a simple communication architecture and a distributed controller that respects information constraints while achieving good closed-loop performance}? 
Addressing this question requires (i) an explicit characterization of how a communication graph constrains the SLS system responses, (ii) a tractable formulation that avoids the combinatorial nature of graph selection, and (iii) a scalable solver that exploits the localized structure of the network.

\begin{figure}[t]
    \centering
    \includegraphics[width=1.0\linewidth]{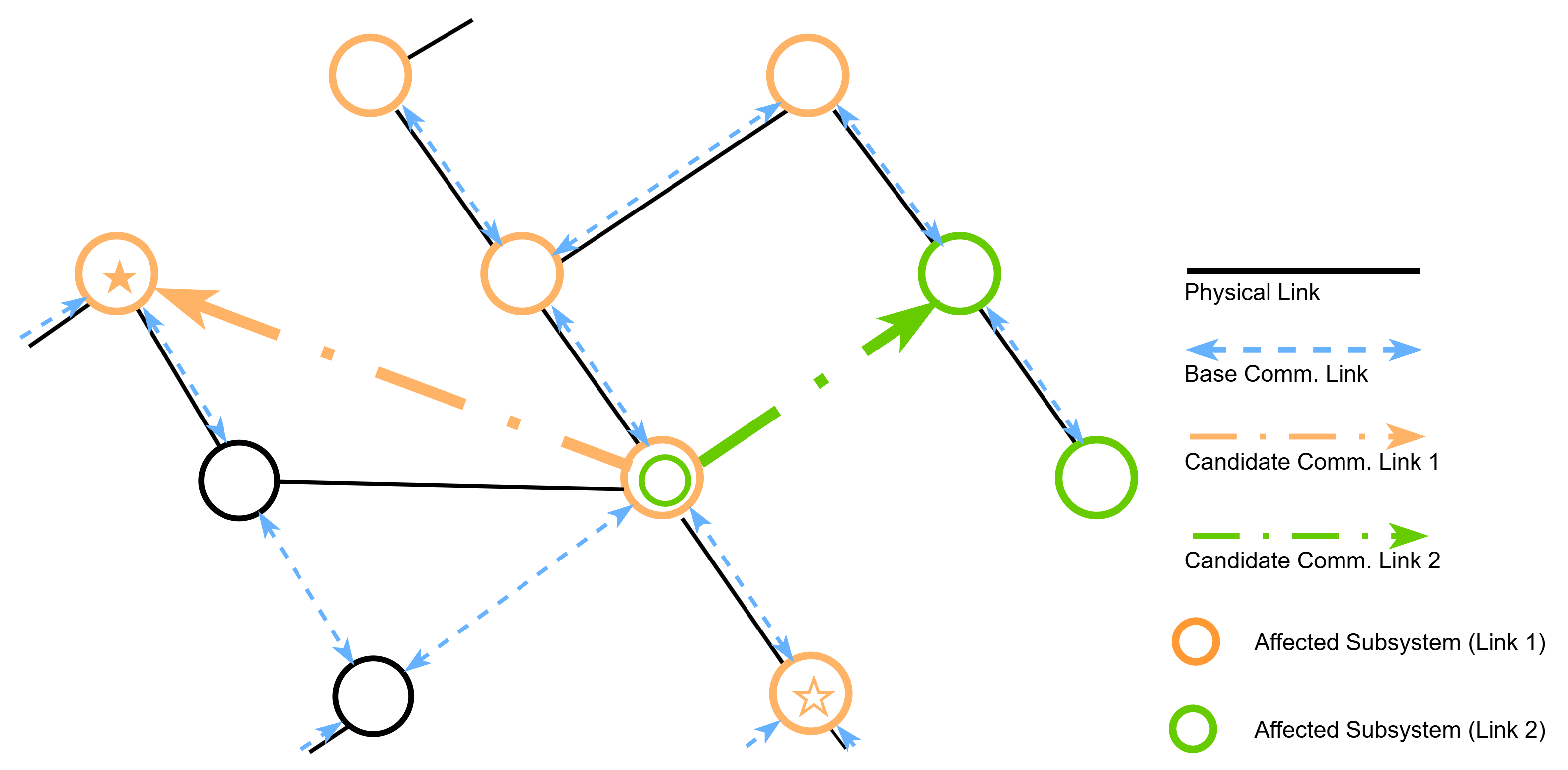}
    \caption{Schematic overview of the proposed 
    framework. Added Comm link 1 (orange line) reduces the path distance from the hollow star subsystem to the solid star subsystem from 4 to 3}
    \label{fig:schematic_overview}
\end{figure}

In this paper, we develop a convex co-design approach within the SLS framework. 
First, we decouple the physical topology from the communication topology and characterize a communication-induced feasible subspace of system responses parameterized by a binary adjacency matrix $\Gamma$ and a communication propagation speed $v_c$.
We show that a coverage condition relating plant interactions to communication reachability is necessary and sufficient for this characterization.
Second, following the regularization for design paradigm~\cite{matni2016regularization} and the atomic norm construction for communication architecture design~\cite{matni2015communication}, we relax the combinatorial topology selection problem into a convex regularized SLS program that promotes sparse use of additional communication links on top of a base graph.
Third, we solve the resulting problem using a proximal ADMM decomposition. This approach preserves separability of the SLS update and parallelizes the architecture update across candidate links by decoupling the quadratic penalty via a proximal term.
We demonstrate the proposed method on unstable networked topologies, and show that the proposed convex relaxation identifies high-value communication links more effectively than standard search baselines under comparable computational budgets.

The organization of this paper is as follows.
After introducing the necessary preliminaries in Section~\ref{sec:preliminaries}, 
Section~\ref{sec:problem formulation} explicitly characterizes the communication-induced system response subspace and establishes the convex regularized SLS formulation for communication architecture co-design problem.
To solve this in scalable, distributed, and localized manner, Section~\ref{sec:comm_cod_alg} develops an algorithm leveraging proximal ADMM. Section~\ref{sec:EXPERIMENTS} then provides numerical experiments to validate the performance and scalability of the proposed framework in optimizing the trade-off between closed-loop performance and communication architecture complexity. We conclude the paper in Section~\ref{sec:conclusion}.

\section{PRELIMINARIES} \label{sec:preliminaries}

\subsection{Interconnected System Model} \label{subsec:system_model}
In this paper, we consider distributed discrete-time LTI systems composed of $N_s$ subsystems, indexed by the node set $\mathscr{V} {\,\coloneqq\,} \{1, \dots ,N_s\}$. 
The system's interconnection is characterized by two directed graphs: the physical graph $\mathscr{G}_{\mathrm{phy}}{\,\coloneqq\,}(\mathscr{V}, \mathscr{E}_{\mathrm{phy}})$ and a communication graph $\mathscr{G}_{\mathrm{comm}}{\,\coloneqq\,}(\mathscr{V}, \mathscr{E}_{\mathrm{comm}})$. An edge $(i, j) \in \mathscr{E}_{\mathrm{phy}}$ indicates that the state of subsystem $j$ directly affects the state evolution of subsystem $i$, whereas an edge $(i, j) \in \mathscr{E}_{\mathrm{comm}}$ indicates that subsystem $j$ can directly transmit information to subsystem $i$.
We define the physical and communication adjacency matrices, $\bar{A}, \Gamma \in \{0,1\}^{N_s \times N_s}$, such that $\bar{A}_{ij}=1$ iff $(i,j) \in \mathscr{E}_{\mathrm{phy}}$, and $\Gamma_{ij}=1$ iff $(i,j) \in \mathscr{E}_{\mathrm{comm}}$. 
For an adjacency matrix $S \in \{\bar{A}, \Gamma\}$, we write $S^r$ for its Boolean $r$-th power, where $(S^r)_{ij}=1$ if and only if there exists a directed path of length at most $r$ from $j$ to $i$ on the respective graph. Equivalently, we denote this topological distance as $\mathrm{dist}_S(j \to i) \leq r$.

Let $\mathscr{N}_i = \{j \mid (i, j) \in \mathscr{E}_{\mathrm{phy}}\}$ denote the set of incoming neighbors for subsystem $i$.
The local dynamics of subsystem $i$ can then be written as:
\begin{equation}
    x_i[k+1] = A_{ii}x_i[k] + B_{ii}u_i[k] + \sum_{j \in \mathscr{N}_i} A_{ij}x_j[k] + w_i[k]
\end{equation}
where $x_i \in \mathbb{R}^n$, $u_i \in \mathbb{R}^m$, and $w_i \in \mathbb{R}^n$ denote the subsystem state, input, and exogenous disturbance, respectively. 
By aggregating the local dynamics, stacking all subsystems together yields the global model:
\begin{equation}
\label{eq:global_plant_model}
    x[k+1] = Ax[k] + B_2u[k] + w[k]
\end{equation}
where $x$, $u$, and $w$ are the stacked vectors of all local states, control actions, and disturbances, respectively. 

\subsection{System Level Synthesis} \label{subsec:system_level_synthesis}
We adopt the same notation conventions as used in standard SLS literature~\cite{wang2014localized,wang2016localized,anderson2019system}.
Under state feedback, SLS parametrizes the closed-loop responses from disturbance to state and control as
$\mathbf{x} = \mathbf{\Phi}_x \mathbf{w}, \mathbf{u} = \mathbf{\Phi}_u \mathbf{w}$.
According to the fundamental theorem of SLS \cite[Theorem 4.1]{anderson2019system}, these system responses are achievable by an internally stabilizing state-feedback controller if and only if
\begin{equation}
\label{eq:sls_achievability}
    \begin{bmatrix} zI - A & -B_2 \end{bmatrix}
    \begin{bmatrix} \mathbf{\Phi}_x \\ \mathbf{\Phi}_u \end{bmatrix}
    = I,
    \quad
    \mathbf{\Phi}_x,\mathbf{\Phi}_u \in \tfrac{1}{z}\mathcal{RH}_\infty
\end{equation}

One can additionally impose spatio-temporal constraints
\begin{equation}\label{eq:sls_const_dist_locality}
        \boldsymbol{\Phi}\coloneqq \begin{bmatrix} \mathbf{\Phi}_x \\ \mathbf{\Phi}_u \end{bmatrix} \in \mathcal{D} \cap \mathcal{L}_d \cap \mathcal{F}_T,
\end{equation}
where 
$\mathcal{D}$ encodes distributed information constraints (to be detailed later),
$\mathcal{L}_d$ encodes spatial locality over $d$-hop neighborhoods of the physical plant,
and $\mathcal{F}_T := \{ \mathbf{G}(z) \in \mathcal{RH}_{\infty} | \mathbf{G}(z)=\sum_{t=1}^T z^{-t}G[t] \}$ denotes a finite-impulse-response (FIR) constraint of horizon $T$.

If communication over a graph $\Gamma$ propagates with a constant speed $v_c{\,\geq\,}1$ (hops per time step) over links assumed ideal, i.e., lossless
and free of latency beyond propagation,
the maximum topological distance information can travel within $\tau$ steps is bounded by $\lfloor v_c \tau \rfloor$, and so $(\Phi_u)_{ij}[\tau+1]$ can be non-zero if and only if $\mathrm{dist}_{\Gamma}(j \to i) \le \lfloor v_c \tau \rfloor$. 
By letting $t{\,=\,}\tau + 1$, this communication delay constraint on control response $\mathbf{\Phi}_u$ translates into a sparsity constraint on the spectral components of $\mathbf{\Phi}_u$:
\begin{equation} \label{eq:comm_delay_constraint}
    \mathrm{supp}(\Phi_u[t]) \subseteq \Gamma^{\lfloor v_c(t-1) \rfloor} \otimes \mathbf{1}_{m\times n}, \quad \forall t \geq 1,
\end{equation}
where, $\mathbf{1}_{m\times n} \in \{0,1\}^{m\times n}$ is the all-ones matrix.
This explicit relation provides the foundational basis for formulating the communication topology-induced feasible subspace of system responses in Section \ref{subsec:feasible_subspace}.

\section{PROBLEM FORMULATION} \label{sec:problem formulation}

In this section, we formulate a convex regularized System Level Synthesis (SLS) problem to co-design the communication architecture.

\subsection{\texorpdfstring{$\Gamma$}{Gamma}-induced Feasible Subspace of System Responses}
\label{subsec:feasible_subspace}

Recall that we explicitly distinguish between the communication topology, denoted by a binary adjacency matrix $\Gamma$, and the physical topology, represented by $\bar{A}$. 
To formalize how the communication graph induces a feasible spatio-temporal constraint on the SLS system responses, we first define a compatibility condition between the physical and communication topologies.

\begin{definition}[\textbf{Coverage Condition}]
\label{def:comm_coverage_cond}
Given the physical topology support $\bar{A}$ and the communication speed $v_c \geq 1$,
a communication topology $\Gamma$ is said to satisfy the \textit{communication coverage condition} if $\bar{A} \subseteq \Gamma^{\lfloor v_c \rfloor}.$
\end{definition}
    
Definition~\ref{def:comm_coverage_cond} enforces that every direct physical dependency encoded by $\bar{A}$ can be communicated across $\Gamma$ within one plant time step.
We now show that this condition exactly characterizes the support structure
induced by $\Gamma$.


\begin{proposition}
\label{prop:induce_Phi_x_subspace_from_comm}
Given system~\eqref{eq:global_plant_model}, let $\Gamma$ be a communication adjacency matrix and $\bar{A}$ the physical topology support. Assume $v_c \ge 1$ and suppose $B_2$ has a block-diagonal structure, i.e., $\mathrm{supp}(B_2) \subseteq I_{N_s} \otimes \mathbf{1}_{n \times m}$. Let $\mathbf{\Phi}$ satisfy the achievability constraint~\eqref{eq:sls_achievability} and the communication delay constraint~\eqref{eq:comm_delay_constraint} on $\mathbf{\Phi}_u$.
Then $\Gamma$ satisfies the coverage condition in Definition~\ref{def:comm_coverage_cond} if and only if
\begin{equation}
\label{eq:phi_x_support}
\mathrm{supp}(\mathbf{\Phi}_x[t]) \subseteq \Gamma^{\lfloor v_c(t-1)\rfloor}
\otimes \mathbf{1}_{n\times n}, \quad \forall t \ge 1.
\end{equation}
\end{proposition}

\begin{proof}
($\Rightarrow$) Using the time-domain form of \eqref{eq:sls_achievability}, the SLS responses satisfy $\Phi_x[1]=I$ and $\Phi_x[t+1]=A\Phi_x[t]+B_2\Phi_u[t]$ for $t\geq 1$.
We proceed by induction.
For $t=1$, we get $\Phi_x[1] = I$,
and so $\mathrm{supp}(\Phi_x[1]) = \mathrm{supp}(I) \subseteq \Gamma^0 \otimes \mathbf{1}_{n\times n}$.
Now, assume the hypothesis holds for time $t\geq 1$.
Then the support of the state response at $t+1$ is bounded:
\begin{align*}
    & \mathrm{supp}(\Phi_x[t+1]) = \mathrm{supp}(A\Phi_x[t] + B_2 \Phi_u[t]) \\
    &\qquad\qquad \subseteq \mathrm{supp}(A)\mathrm{supp}(\Phi_x[t]) \cup \mathrm{supp}(B_2) \mathrm{supp}(\Phi_u[t])
\end{align*}
Recall that $B_2$ has a block-diagonal structure, i.e., $\mathrm{supp}(B_2)\subseteq I_{N_s} \otimes \mathbf{1}_{n \times m}$.
By applying~\eqref{eq:comm_delay_constraint} and  the mixed-product property of Kronecker product,
we have 
\begin{align*}
    \mathrm{supp}(B_2)\mathrm{supp}(\Phi_u[t])     
    & \subseteq (I_{N_s} \otimes \mathbf{1}_{n \times m}) (\Gamma^{\lfloor v_c(t-1) \rfloor} \otimes \mathbf{1}_{m\times n}) \\
    & = \Gamma^{\lfloor v_c(t-1) \rfloor} \otimes \mathbf{1}_{n\times n}
\end{align*}
Together with the coverage condition Definition~\ref{def:comm_coverage_cond},
we obtain:
\begin{align*}
    &\mathrm{supp}(\Phi_x[t+1]) \\
    &\subseteq \left((\bar{A} \;\Gamma^{\lfloor v_c(t-1) \rfloor})\otimes \mathbf{1}_{n \times n}\right) \cup \left(\Gamma^{\lfloor v_c(t-1) \rfloor} \otimes \mathbf{1}_{n\times n}\right) \\
    &\subseteq (\Gamma^{\lfloor v_c \rfloor}\,\Gamma^{\lfloor v_c(t-1) \rfloor} \cup \Gamma^{\lfloor v_c(t-1) \rfloor}) \otimes \mathbf{1}_{n\times n} \\
    &\subseteq \Gamma^{\lfloor v_c t \rfloor}\otimes \mathbf{1}_{n\times n}.
\end{align*}

($\Leftarrow$) We prove the contrapositive. Suppose $\bar{A} \not\subseteq \Gamma^{\lfloor v_c \rfloor}$, and pick $(i,j)$ with $\bar{A}_{ij} = 1$ and $(\Gamma^{\lfloor v_c \rfloor})_{ij} = 0$. Since $\Gamma^{r} \supseteq I_{N_s}$ for all $r \ge 0$, necessarily $i \neq j$. From $\mathbf{\Phi}_x[1] = I$ we obtain $\mathbf{\Phi}_x[2] = A + B_2\mathbf{\Phi}_u[1]$, while \eqref{eq:comm_delay_constraint} at $t=1$ gives $\mathrm{supp} (\mathbf{\Phi}_u[1]) \subseteq I_{N_s} \otimes \mathbf{1}_{m\times n}$, so $B_2\mathbf{\Phi}_u[1]$ is block diagonal.
Its $(i,j)$ block therefore vanishes, so the $(i,j)$ block of $\mathbf{\Phi}_x[2]$ equals $A_{ij}$, which is nonzero since $\bar{A}_{ij} = 1$,
so \eqref{eq:phi_x_support} fails at $t = 2$. \hfill\qedhere
\end{proof}


Proposition~\ref{prop:induce_Phi_x_subspace_from_comm} shows that the coverage
condition is exactly what is required for the communication topology $\Gamma$
to induce the following admissible subspace of system responses:
\begin{equation}
\label{subspace_dist}
        \mathcal{D}(\Gamma) := \bigoplus_{t \geq 1} \frac{1}{z^t}\mathcal{D}^{(t)}(\Gamma)
\end{equation}    
where each term is defined as:
\begin{align}
\label{subspace_dist_time}
    &\mathcal{D}^{(t)}(\Gamma) =\\& \left\{ \begin{bmatrix} \Phi_x[t] \\ \Phi_u[t] \end{bmatrix} \;\middle|\; 
    \begin{aligned}
    &\mathrm{supp}(\Phi_x[t]) \subseteq \Gamma^{\lfloor v_c(t-1) \rfloor}\otimes \mathbf{1}_{n \times n}, \\
    &\mathrm{supp}(B_2\Phi_u[t]) \subseteq \Gamma^{\lfloor v_c(t-1) \rfloor} \otimes \mathbf{1}_{n \times n}
    \end{aligned} \right\} \nonumber
\end{align}
Combining this $\Gamma$-induced constraint with locality, the spatio-temporal constraints~\eqref{eq:sls_const_dist_locality} can be expressed as

\begin{equation}
\label{eq:feasible_spatio_temporal_constraint}
    \mathcal{S}(\Gamma) := \mathcal{D}(\Gamma) \cap \mathcal{L}_d \cap \mathcal{F}_T 
\end{equation}

The full SLS problem considering communication topology satisfying the coverage condition can be written as follows, with~\eqref{eq:sls_const_dist_locality} replaced by~\eqref{eq:feasible_spatio_temporal_constraint}:
\begin{subequations}
\label{prob:sls_problem_with_Gamma}
\begin{align}
\min_{\mathbf{\Phi}} \quad 
    & \left\| \begin{bmatrix} C_1 & D_{12} \end{bmatrix}
    \mathbf{\Phi}
            \right\|_{\mathcal{H}_2}^2 
            \label{eq:sls_obj}\\
\text{s.t.}\quad 
    & \text{achievability }\eqref{eq:sls_achievability}\\
    & \mathbf{\Phi}
    \in \mathcal{S}(\Gamma) \label{eq:sls_feasible_const_with_Gamma}
\end{align}
\end{subequations}

\subsection{SLS Problem for Communication Graph Co-Design}

Building upon the feasible system response subspace characterization in Section \ref{subsec:feasible_subspace}, we now formulate the SLS problem for communication graph co-design by leveraging the \textit{regularization for design (RFD)} framework \cite{matni2016regularization,matni2015communication}. We consider a practical design scenario where the communication topology $\Gamma$ is constructed by augmenting a fixed, mandatory base topology $\Gamma_{\mathrm{base}}$ with a limited number of selected additional links $\mathscr{E}_{\text{sel}}$ from a set of candidate edges $\mathscr{E}$.
Here, the base topology $\Gamma_{\mathrm{base}}$ should satisfy Definition~\ref{def:comm_coverage_cond}; one practical choice is the physical topology $\bar{A}$.

To design a distributed controller that uses at most $l$ additional communication links, we impose a cardinality constraint on the augmented links.
From Problem \eqref{prob:sls_problem_with_Gamma}, the communication link constrained optimization control problem is formulated as follows
\begin{subequations}
\label{prob:sls_comm_codesign_combinatorial}
\begin{align}
\min_{\mathbf{\Phi},\Gamma}\quad
    & \left\|
            \begin{bmatrix} C_1 & D_{12} \end{bmatrix}
            \mathbf{\Phi}
       \right\|_{\mathcal{H}_2}^2
\label{eq:comb_obj}\\
\text{s.t.} \quad
    & \text{achievability }(\ref{eq:sls_achievability}) \\
    & \mathbf{\Phi}
    \in \mathcal{S}(\Gamma)
    \label{eq:sls_dist_localized_constraint} \\
    & \mathscr{E}_{\text{sel}} \subseteq \mathscr{E},\hskip.3cm 
    \mathrm{card}\!\left(\mathscr{E}_{\text{sel}}\right) \leq l \label{eq:comb_cardinality}
\end{align}
\end{subequations}
where $l \in \mathbb{Z}_{\geq 0}$ is the design budget.
Due to the cardinality constraint \eqref{eq:comb_cardinality}, Problem \eqref{prob:sls_comm_codesign_combinatorial} is a combinatorial problem which becomes computationally intractable, especially for large-scale systems.
To address this, we relax this constraint by removing the topology-dependent subspace $\mathcal{D}(\Gamma)$ in~\eqref{eq:feasible_spatio_temporal_constraint},
and introduce a convex penalty term $\|\mathbf{\Phi}_u\|_{\mathrm{comm}}$, to penalize the use of additional communication links.
The relaxed SLS problem with this regularization is formulated as follows:
\begin{subequations}
\label{prob:sls_comm_codesign_basic_structure}
\begin{align}
\min_{\mathbf{\Phi}} \quad
    & \left\| \begin{bmatrix} C_1 & D_{12} \end{bmatrix} \mathbf{\Phi} \right\|_{\mathcal{H}_2}^2
            \;+\; \lambda \|\mathbf{\Phi}_u\|_{\mathrm{comm}}  \label{eq:sls_comm_codesign_basic_structure_obj}\\
\text{s.t.} \quad
    & \text{achievability }(\ref{eq:sls_achievability})
    \label{eq:sls_comm_codesign_basic_structure_achievability}\\    
    & \mathbf{\Phi} \in \mathcal{L}_d \cap \mathcal{F}_T \label{eq:sls_comm_codesign_spatio_temporal_constraint}
\end{align}
\end{subequations}

Note that the solution to Problem \eqref{prob:sls_comm_codesign_basic_structure} should be realizable as a distributed controller.
To this end, we define distributed realizability as follows.

\begin{definition}[\textbf{Distributed Realizability}]
\label{def:dist_realizable}
Given an augmented communication topology $\Gamma$, a system response $\mathbf{\Phi}$ is said to be \textit{distributed realizable on $\Gamma$} if it satisfies the achievability constraint \eqref{eq:sls_achievability} and $\mathbf{\Phi} \in \mathcal{S}(\Gamma)$ as in~\eqref{eq:feasible_spatio_temporal_constraint}.
\end{definition}


The penalty in~\eqref{eq:sls_comm_codesign_basic_structure_obj} acts on $\mathbf{\Phi}_u$ alone, 
since by the achievability constraint~\eqref{eq:sls_achievability} the state response $\mathbf{\Phi}_x$ is determined by $\mathbf{\Phi}_u$. Penalizing $\mathbf{\Phi}_x$ would therefore be redundant, and would introduce a bias unrelated to the communication cost that distorts the link selection.
To characterize the feasible subspace of $\mathbf{\Phi}_u$, let $\mathcal{D}_u(\Gamma):=\mathcal{P}_u(\mathcal{D}(\Gamma))$ and $\mathcal{L}_{d,u}:=\mathcal{P}_u(\mathcal{L}_d)$, where $\mathcal{P}_u$ denotes the projection operator onto the control response space (i.e., $\mathbf{\Phi}_u = \mathcal{P}_u(\mathbf{\Phi})$). We define $\mathcal{S}_u(\Gamma) := \mathcal{D}_u(\Gamma) \cap \mathcal{L}_{d,u} \cap \mathcal{F}_T$ as the subspace of finite-horizon \emph{control} responses compatible with topology $\Gamma$, analogous to~\eqref{eq:feasible_spatio_temporal_constraint}. This restriction incurs no loss of generality, as formalized in the following remark.

\begin{remark}\label{remark:comm_link_norm}
    Let $\mathbf{\Phi}^\star$
    be the optimal solution to Problem \ref{prob:sls_comm_codesign_basic_structure}. Suppose the communication topology $\Gamma^\star$, induced by the optimal control response $\mathbf{\Phi}_u^\star \in \mathcal{S}_u(\Gamma^\star)$, satisfies the coverage condition in Definition \ref{def:comm_coverage_cond}. 
    Then, by Proposition \ref{prop:induce_Phi_x_subspace_from_comm}, the system response $\mathbf{\Phi}^\star$ is \textit{distributed realizable on $\Gamma^\star$} in the sense of Definition~\ref{def:dist_realizable}.
\end{remark}

\subsection{Communication-Link Regularization}
Having established the overall co-design objective in \eqref{eq:sls_comm_codesign_basic_structure_obj}, we now detail the explicit construction of the communication regularizer $\|\mathbf{\Phi}_u\|_{\mathrm{comm}}$.
We first define the atomic subspaces that capture the contribution of individual links.
To ensure that these atomic components inherently satisfy the spatio-temporal constraints of the control response, we confine the ambient space to $\mathcal{U} := \mathcal{L}_{d,u} \cap \mathcal{F}_T$.
Within this space, equipped with the Frobenius inner product, the atomic subspace associated with adding a candidate link $(i,j) \in \mathscr{E}$ to the base topology $\Gamma_{\mathrm{base}}$ is defined as:
\begin{equation}
    \label{eq:Eij_subspace}
    \mathcal{E}_{ij} := \mathcal{S}_u(\Gamma_{\mathrm{base}})^{\perp_{\mathcal{U}}} \cap \mathcal{S}_u(\Gamma_{\mathrm{base}} \cup E_{ij}),
\end{equation}
where $E_{ij} \in \{0, 1\}^{N_s \times N_s}$ is a binary matrix with its $(i, j)$-th entry equal to $1$ and $0$ elsewhere (representing the candidate directed link $(i,j)$), and $\perp_\mathcal{U}$ denotes the orthogonal complement within $\mathcal{U}$.
Geometrically, $\mathcal{E}_{ij}$ represents the component of the control-response subspace newly enabled by adding link $(i,j)$, orthogonal to what was already achievable with base topology $\Gamma_{\mathrm{base}}$.

With these atomic subspaces defined, we can now express the overall control response $\mathbf{\Phi}_u$ as a sum of a base component $\mathbf{\Phi}_{u,\mathrm{base}} \in \mathcal{S}_u(\Gamma_{\mathrm{base}})$ and individual atomic components $\mathbf{\Phi}_{u,ij} \in \mathcal{E}_{ij}$, i.e., 
\[
    \mathbf{\Phi}_u = \mathbf{\Phi}_{u,\mathrm{base}} + \sum_{(i,j) \in \mathscr{E}} \mathbf{\Phi}_{u,ij}.
\]
Based on this decomposition, we define the communication regularizer $\|\mathbf{\Phi}_u\|_{\mathrm{comm}}$ as the optimal value of the problem:
\begin{align}
\label{prob:comm_link_norm}
\|\mathbf{\Phi}_u\|_{\mathrm{comm}} \;:=\;
    \min_{\substack{\mathbf{\Phi}_{u,\mathrm{base}},\\ \{\mathbf{\Phi}_{u,ij}\}}} \quad
        & \sum_{(i,j)\in\mathscr{E}} \|\mathbf{\Phi}_{u,ij}\|_{\mathcal{H}_2} \\
    \text{s.t.}\quad
        & \mathbf{\Phi}_u = \mathbf{\Phi}_{u,\mathrm{base}} + \sum_{(i,j)\in\mathscr{E}} \mathbf{\Phi}_{u,ij}, \nonumber\\
        & \mathbf{\Phi}_{u,\mathrm{base}} \in \mathcal{S}_u(\Gamma_{\mathrm{base}}), \nonumber\\
        & \mathbf{\Phi}_{u,ij} \in \mathcal{E}_{ij}, \qquad \forall (i,j)\in\mathscr{E}. \nonumber
\end{align}
Essentially, one can interpret it as the atomic gauge induced by the base subspace and the single-link augmentation subspaces.
By substituting this definition into the relaxed SLS formulation \eqref{prob:sls_comm_codesign_basic_structure}, we obtain the final convex program for control and communication topology co-design:
\begin{subequations}
\label{prob:sls_comm_codesign}
\begin{align}
\min_{\substack{\mathbf{\Phi},\mathbf{\Phi}_{\mathrm{u,base}},\\ \{\mathbf{\Phi}_{u,ij}\}}} \quad
    & \left\| \begin{bmatrix} C_1 & D_{12} \end{bmatrix} \mathbf{\Phi} \right\|_{\mathcal{H}_2}^2
            + \lambda \sum_{(i,j)\in\mathscr{E}} \|\mathbf{\Phi}_{u,ij}\|_{\mathcal{H}_2} \label{eq:sls_comm_codesign_obj}\\
\text{s.t.} \quad
    & \begin{bmatrix} zI-A & -B_2 \end{bmatrix} \mathbf{\Phi} = I, \\
    & 
    \mathbf{\Phi} \in \mathcal{L}_d \cap \mathcal{F}_T \\
    & \mathbf{\Phi}_u = \mathbf{\Phi}_{u,\mathrm{base}} + \sum_{(i,j)\in\mathscr{E}}\mathbf{\Phi}_{u,ij}, \label{eq:sls_comm_codesign_decomp}\\
    & \mathbf{\Phi}_{u,\mathrm{base}} \in \mathcal{S}_u(\Gamma_{\mathrm{base}}), \nonumber \\
    &\mathbf{\Phi}_{u,ij}\in\mathcal{E}_{ij} \quad \forall(i,j)\in\mathscr{E} \label{eq:sls_comm_codesign_atom_subspaces}
\end{align}
\end{subequations}
The quadratic $\mathcal{H}_2$ performance objective~\eqref{eq:sls_comm_codesign_obj} features a $\ell_1/\ell_2$ regularization term. When $\lambda = 1$, this is equivalent to the group lasso penalty widely used in the statistical learning literature \cite{yuan2006model}. This regularizer induces block-sparsity in the decomposition of $\mathbf{\Phi}_u$, effectively driving entire atomic components $\mathbf{\Phi}_{u,ij}$ to zero unless the addition of the corresponding link $(i,j)$ yields a performance improvement that outweighs the communication penalty.
Consequently, the weighting parameter $\lambda{\,>\,}0$ explicitly governs the trade-off between the closed-loop performance and the complexity of the designed communication architecture.

To address the validity of the communication topology defined by solving Problem~\eqref{prob:sls_comm_codesign}, we first state the following intermediate result.

\begin{lemma}
\label{lem:monotone_Su}
If $\Gamma_1 \subseteq \Gamma_2$, then $\mathcal{S}_u(\Gamma_1)\subseteq \mathcal{S}_u(\Gamma_2)$.
\end{lemma}

The proof of this is simple, so we leave it as a sketch.
The inclusion $\Gamma_1\subseteq \Gamma_2$ implies $\Gamma_1^{\lfloor v_c(t-1)\rfloor}\subseteq \Gamma_2^{\lfloor v_c(t-1)\rfloor}$ for all $t\ge 1$.
Hence $\mathcal{D}_u(\Gamma_1)\subseteq \mathcal{D}_u(\Gamma_2)$, and intersecting both sides with the fixed subspace $\mathcal{L}_{d,u}\cap \mathcal{F}_T$ yields the claim.
Lemma~\ref{lem:monotone_Su} intuitively says that adding communication links enlarges the feasible subspace, ensuring that the optimal $\mathcal{H}_2$ cost is monotonically non-increasing.

\begin{proposition} \label{prop:sls_dist_realizable}
Assume base topology $\Gamma_{\mathrm{base}}$ satisfies the coverage condition in Definition~\ref{def:comm_coverage_cond}, and let $\mathbf{\Phi}^*$ be the optimal solution to Problem \ref{prob:sls_comm_codesign} with corresponding optimally selected link set $\mathscr{E}_{\mathrm{sel}}^* \coloneqq \{(i ,j) \in \mathscr{E} \mid \mathbf{\Phi}_{u,ij}^* \neq 0\}$ and adjacency matrix $\Gamma_{\mathrm{des}} \coloneqq \Gamma_{\mathrm{base}} \cup (\cup_{(i,j)\in\mathscr{E}_{\mathrm{sel}}^*} E_{ij})$.
Then $\mathbf{\Phi}^*$ is distributed realizable on $\Gamma_{\mathrm{des}}$
in the sense of Definition~\ref{def:dist_realizable}.
\end{proposition}

\begin{proof}
Since $\mathbf{\Phi}^*$ satisfies decomposition constraints \eqref{eq:sls_comm_codesign_decomp} and \eqref{eq:sls_comm_codesign_atom_subspaces}, we have $\mathbf{\Phi}_u^* = \mathbf{\Phi}_{u,\mathrm{base}}^* + \sum_{(i,j)\in\mathscr{E}_{\mathrm{sel}}^*} \mathbf{\Phi}_{u,ij}^*$, so $\mathbf{\Phi}_{u,\mathrm{base}}^*\in \mathcal{S}_u(\Gamma_{\mathrm{base}})$ and $\mathbf{\Phi}_{u,ij}^*\in \mathcal{S}_u(\Gamma_{\mathrm{base}}\cup E_{ij})$.
Here, we add superscript $*$ to denote the variables in Problem~\ref{prob:sls_comm_codesign} corresponding to $\mathbf{\Phi}^*$.
Since $\Gamma_{\mathrm{base}}\subseteq \Gamma_{\text{des}}$ and
$\Gamma_{\mathrm{base}}\cup E_{ij}\subseteq \Gamma_{\text{des}}$ for every $(i,j)\in\mathscr{E}_{\mathrm{sel}}^*$, Lemma~\ref{lem:monotone_Su} gives $\mathcal{S}_u(\Gamma_{\mathrm{base}})\subseteq \mathcal{S}_u(\Gamma_{\text{des}})$,and $\mathcal{S}_u(\Gamma_{\mathrm{base}}\cup E_{ij})\subseteq \mathcal{S}_u(\Gamma_{\text{des}})$.
Because $\mathcal{S}_u(\Gamma_{\text{des}})$ is a linear subspace, it follows that $\mathbf{\Phi}_u^* \in \mathcal{S}_u(\Gamma_{\text{des}})$.
Finally, $\Gamma_{\text{des}}$ contains the covering graph $\Gamma_{\mathrm{base}}$, so it also satisfies the coverage condition. 
Proposition~\ref{prop:induce_Phi_x_subspace_from_comm} implies $\mathbf{\Phi}^* \in \mathcal{S}(\Gamma_{\text{des}})$, and therefore $\mathbf{\Phi}^*$ is distributed realizable on $\Gamma_{\text{des}}$.
\end{proof}

\section{COMMUNICATION ARCHITECTURE CO-DESIGN ALGORITHM} \label{sec:comm_cod_alg}

We now describe a scalable solver for Problem~\eqref{prob:sls_comm_codesign}.
The closed-loop $\mathcal{H}_2$ objective is separable across block-columns under the locality/FIR constraints, while the communication regularizer is separable across candidate links but coupled through the decomposition constraint.
To exploit this structure, we apply proximal ADMM~\cite{fazel2013hankel} to split the controller-synthesis variables from the architecture (link-selection) variables, and we parallelize the architecture update by adding an  appropriate proximal term.

\subsection{Proximal ADMM Splitting and Parallel Updates} \label{subsec:admm_split_and_parallel_updates}
We introduce the following proximal-ADMM consensus splitting of Problem~\ref{prob:sls_comm_codesign}:
\begin{subequations}
\label{prob:admm_sls_comm_codesign}
\begin{align}
\min_{\mathbf{\Phi}_x,\;\mathbf{\Phi}_{u1},\;\mathbf{\Phi}_{u2}} \quad
    & f(\mathbf{\Phi}_x, \mathbf{\Phi}_{u1}) \; + \; g(\mathbf{\Phi}_{u2}) \label{eq:admm_sls_comm_codesign_obj}\\
\text{s.t.}\quad
    & \mathbf{\Phi}_{u1} = \mathbf{\Phi}_{u2} \label{eq:admm_sls_comm_codesign_constraint}
\end{align}
\end{subequations}
where
\begin{align}
    &f(\mathbf{\Phi}_x, \mathbf{\Phi}_{u1}) :=\\
    &
        \begin{cases}
        \left\| \begin{bmatrix} C_1 \; D_{12} \end{bmatrix}\begin{bmatrix} \mathbf{\Phi}_x \\ \mathbf{\Phi}_{u1} \end{bmatrix} \right\|_{\mathcal{H}_2}^2 
        & \begin{aligned} 
            &\text{if \eqref{eq:sls_achievability} and} 
            \begin{bmatrix}\mathbf{\Phi}_x\\\mathbf{\Phi}_{u1}\end{bmatrix}\in\mathcal{S}(\Gamma_{\max})
          \end{aligned} \\        
        \infty, &\text{otherwise}
        \end{cases}\notag
\end{align}
and $g(\mathbf{\Phi}_{u2}) := \lambda \|\mathbf{\Phi}_{u2}\|_{\mathrm{comm}}$. Here $\Gamma_{\max} \coloneqq \Gamma_{\mathrm{base}} \cup (\cup_{(i,j)\in \mathscr{E}}E_{ij})$ denotes the communication adjacency matrix that includes all candidate links.
To apply ADMM, we introduce an auxiliary variable $\mathbf{\Phi}_{u2}$ and assign the atomic norm regularizer to this variable, while retaining the SLS synthesis constraints on the original system responses.
The equivalence to the original problem is preserved by imposing the linear consensus constraint $\mathbf{\Phi}_{u1}=\mathbf{\Phi}_{u2}$, yielding a variable-splitting constrained formulation.
By \eqref{prob:comm_link_norm}, $g$ is evaluated through $\mathbf{\Phi}_{u2,\mathrm{base}}$ and $\{\mathbf{\Phi}_{u2,ij}\}_{(i,j)\in\mathscr{E}}$.

Our new Problem~\eqref{prob:admm_sls_comm_codesign} is of the standard ADMM form \cite{boyd2011distributed}. By applying proximal ADMM,
it follows the iterations:
\begin{subequations}
\begin{align}
    (\mathbf{\Phi}_{x}^{k+1}, \mathbf{\Phi}_{u1}^{k+1}) &= \arg\min_{\mathbf{\Phi}_x, \mathbf{\Phi}_{u1}} \quad f(\mathbf{\Phi}_x, \mathbf{\Phi}_{u1}) \notag \\
    &\qquad + \frac{\rho}{2} \| \mathbf{\Phi}_{u1} - \mathbf{\Phi}_{u2}^k + \mathbf{\Lambda}^k\|_{\mathcal{H}_2}^2 \label{eq:admm_step1_sls_update} \\[1em]
    \mathbf{\Phi}_{u2}^{k+1} &= \arg\min_{\mathbf{\Phi}_{u2}} \quad g(\mathbf{\Phi}_{u2}) \notag \\
    &\qquad + \frac{\rho}{2} \| \mathbf{\Phi}_{u2} - \mathbf{\Phi}_{u1}^{k+1} - \mathbf{\Lambda}^k\|_{\mathcal{H}_2}^2 \notag \\
    &\qquad + \frac{1}{2} \| \mathbf{\Phi}_{u2} - \mathbf{\Phi}_{u2}^k\|_{P}^2 \label{eq:admm_step2_architecture_update} \\[1em]
    \mathbf{\Lambda}^{k+1} &= \mathbf{\Lambda}^k + \mathbf{\Phi}_{u1}^{k+1} - \mathbf{\Phi}_{u2}^{k+1}. \label{eq:admm_step3_dual_update}
\end{align}
\end{subequations}

Subproblem \eqref{eq:admm_step1_sls_update} has the same structure as the localized SLS synthesis problem, with an additional $\mathcal{H}_2$-quadratic term.
As in \cite{anderson2019system}, it admits a parallel solution across block-columns and can exploit locality-based dimension reduction.
Moreover, since the dual update equation \eqref{eq:admm_step3_dual_update} consists solely of matrix addition, it can be computed locally by each subsystem.

The architecture update~\eqref{eq:admm_step2_architecture_update} can be efficiently solved in parallel with an appropriate proximal term, $\|\cdot\|_P^2$, which guarantees globally-optimal convergence.
To define this proximal term, we first note the following structure. The subspace of the base component $\mathbf{\Phi}_{u2,\mathrm{base}}$ and the subspace of each atomic component $\mathbf{\Phi}_{u2,ij}$ are orthogonal by construction. 
Furthermore, because $\mathbf{\Phi}_{u2,\mathrm{base}}$ is not penalized by the communication norm, its update is simply given by the orthogonal projection onto the base subspace: $\mathbf{\Phi}_{u2, \mathrm{base}}^{k+1}= \mathcal{P}_{\mathrm{base}}(\mathbf{\Phi}_{u1}^{k+1}+\mathbf{\Lambda}^k)$,
where $\mathcal{P}_{\mathrm{base}}(\cdot)$ represents projection operator onto $\mathcal{S}_u(\Gamma_{\mathrm{base}})$.
On the other hand, the atomic components are generally not orthogonal to each other, i.e., for distinct candidate links $(i,j), (p,q) \in \mathscr{E}$, the intersection $\mathcal{E}_{ij} \cap \mathcal{E}_{pq}$ (see~\eqref{eq:Eij_subspace}) may be non-trivial.
Thus,
the quadratic penalty term in \eqref{eq:admm_step2_architecture_update} couples them.
By adding the appropriate proximal term, we can remove this coupling and compute the update in a parallelized manner per link.

First, let $z{\,\coloneqq\,}[z_{ij}]_{(i,j) \in \mathscr{E}}$ where $z_{ij}{\,\in\,}\mathbb{R}^{|\mathscr{E}_{ij}|}$ is the flattened vector of $\mathbf{\Phi}_{u2,ij} \in \mathcal{E}_{ij}$.
By setting $P \coloneqq (\rho + c)I - \rho MM^{\top}$, we can rewrite \eqref{eq:admm_step2_architecture_update} as follows:
\begin{align}
\label{eq:vec_prox_arch_update}
    z^{k+1} &= \arg\min_{z} \quad 
         \lambda \sum_{(i,j)\in \mathscr{E}} \|z_{ij}\| \\
        &\hskip.5cm + \frac{\rho + c}{2} \| z - \left(z^k - \frac{\rho}{\rho+c} M (M^\top z^k - V^{k+1})\right) \|^2 \nonumber    
\end{align}
where $V^{k+1}$ is the vectorized form of $\mathbf{\Phi}_{u1}^{k+1} + \mathbf{\Lambda}^k$ with appropriate dimensions, and $M = [M_{ij}]_{(i,j)\in\mathscr{E}} \in \mathbb{R}^{\left(\sum_{(i,j)\in\mathscr{E}} |\mathcal{E}_{ij}|\right)
\times |\mathcal{S}_{u}(\Gamma_{\max})|}$.
Here, each $M_{ij}$ is a selection matrix which maps the global subspace $\mathcal{S}_{u}(\Gamma_{\max})$ to each element of $\mathcal{E}_{ij}$.

With the above choice of proximal term, the architecture update admits a separable form over candidate links, yielding the parallel subproblems
\begin{subequations}
\label{eq:prox_subproblem2_decomp}
\begin{align}
    \mathbf{\Phi}_{u2,ij}^{k+1} = \arg&\min_{\mathbf{\Phi}_{u2,ij}} \quad 
         \lambda\|\mathbf{\Phi}_{u2,ij}\|_{\mathcal{H}_2} \nonumber \\
        &\quad + \frac{\rho}{2} \| \mathbf{\Phi}_{u2,ij}+\mathbf{\Phi}_{u2,-ij}^{k} - \mathbf{\Phi}_{u1}^{k+1} - \mathbf{\Lambda}^k\|_{\mathcal{H}_2}^2  \nonumber \\
        &\quad + \frac{c}{2}\| \mathbf{\Phi}_{u2,ij}-\mathbf{\Phi}_{u2,ij}^{k}\|_{\mathcal{H}_2}^2 \\[0.5em]
    & \quad \text{s.t.} \quad
         \mathbf{\Phi}_{u2,ij} \in \mathcal{E}_{ij},
\end{align}
\end{subequations}
where
\[
\mathbf{\Phi}_{u2,-ij}^{k} \triangleq \sum_{(p,q)\in\mathscr{I}_{ij}} \mathbf{\Phi}_{u2,pq}^{k}.
\]

Here, $\mathscr{I}_{ij}$ denotes the coupling link set for link $(i,j)$, which is defined as the set of other candidate links whose atomic subspaces intersect with $\mathcal{E}_{ij}$: 
\[
\mathscr{I}_{ij} := \{ (p,q) \in \mathscr{E}\setminus\{(i,j)\}\mid \mathcal{E}_{ij}\cap\mathcal{E}_{pq}\neq\{\mathbf{0}\}\}.
\]

To satisfy the proximal ADMM convergence conditions established in \cite{fazel2013hankel}, we choose $c$ such that the proximal matrix satisfies $P \succeq0$, which requires $\rho+c \geq \rho\lambda_{\max}( M^\top M)$. 
Notice that $M^\top M = \sum_{(i,j)\in\mathscr{E}} M_{ij}^\top M_{ij}$. Since each $M_{ij}$ is a selection matrix, $M^\top M$ is a diagonal matrix where each diagonal entry simply counts the number of atomic subspaces that contain the corresponding global basis element. Because any atomic subspace overlaps with at most $|\mathscr{I}_{ij}|$ other candidate links, no single global dimension is shared by more than $\max_{(i,j)} |\mathscr{I}_{ij}| + 1$ subspaces.
Thus, a sufficient condition for $P \succeq 0$ is 
$c \geq \rho \cdot \max_{(i,j)} |\mathscr{I}_{ij}|$.
Furthermore, $\rho \ge 1$ is sufficient for convergence with a unit dual step size.

\subsection{Closed-form Solution for the ADMM Subproblems}
Adopting the notation of \cite{wang2016localized}, we show that the proposed ADMM subproblems admit closed-form solutions.
Since Subproblem \eqref{eq:admm_step1_sls_update}, dimensionally reduced by the $\mathcal{L}_d$ constraint, is a least-squares problem subject to affine constraints, the optimal analytical solution for the $j$-th block-column over its localized neighborhood $o_{j,d}$ can be specified as an affine function:
\begin{equation} \label{eq:closed_form_sls}
    \begin{bmatrix} \mathbf{\Phi}_x^{k+1} \\ \mathbf{\Phi}_{u1}^{k+1} \end{bmatrix}_{o_{j,d}} = F_{o_{j,d}}^a(\mathbf{\Phi}_{u2}^k, \Lambda^k)_{o_{j,d}} + F_{o_{j,d}}^b
\end{equation}
where $F_{o_{j,d}}^a$ and $F_{o_{j,d}}^b$ represent a suitable linear map and an affine term, respectively.
These parameters can be computed by using the vectorization property of the Frobenius norm.

For the per-link architecture update~\eqref{eq:prox_subproblem2_decomp}, let $\mathcal{P}_{\mathcal{E}_{ij}}$ denote the orthogonal projection onto $\mathcal{E}_{ij}$.
By completing the square in the objective function of \eqref{eq:prox_subproblem2_decomp}, the $\ell_2$-norm penalty acts as a group-level regularizer. This shows that the per-link architecture update reduces to a group soft-thresholding (vector shrinkage) step \cite{wright2009sparse}:
\begin{equation}
\label{eq:closed_form_arch}
\mathbf{\Phi}_{u2,ij}^{k+1} =
\left(1-\frac{\lambda/(\rho+c)}{\left\|\mathcal{P}_{\mathcal{E}_{ij}}\!\left(V_{ij}^{k}\right)\right\|_{\mathcal{H}_2}}\right)_{+}\,
\mathcal{P}_{\mathcal{E}_{ij}}\!\left(V_{ij}^{k}\right),
\end{equation}
where $(a)_{+}=\max\{a,0\}$ and
\[
V_{ij}^{k} := \frac{\rho}{\rho+c}\left(\mathbf{\Phi}_{u1}^{k+1}+\Lambda^k-\mathbf{\Phi}_{u2,-ij}^{k}\right) + \frac{c}{\rho+c}\mathbf{\Phi}_{u2,ij}^{k} .
\]

By admitting exact closed-form updates, the algorithm avoids computationally expensive numerical solvers within the ADMM loop. Since the affine parameters in \eqref{eq:closed_form_sls} are computed only once prior to the iterations, the subsequent updates reduce to highly efficient matrix multiplications and soft-thresholding steps.

The full co-design procedure is summarized in Algorithm~\ref{alg:comm-codesign}.
First, from Proposition~\ref{prop:sls_dist_realizable}, we can extract desired communication topology $\Gamma_{\mathrm{des}}$ from the optimal solution to Problem~\ref{prob:sls_comm_codesign}, guaranteeing that the system response is realizable as a distributed controller. Afterwards, we can obtain distributed optimal controller $\mathbf{\Phi}_{\mathrm{des}}^*$ by solving the unregularized Problem \ref{prob:sls_problem_with_Gamma} over $\Gamma_{\mathrm{des}}$ which eliminates any bias and conservatism introduced by the atomic norm regularizer. 
Following \cite{boyd2011distributed}, the algorithm terminates when the primal and dual residuals fall below their respective tolerances.

\subsection{Scalable and Localized Synthesis} 
\label{subsec:dist_impl}

To guarantee the overall scalability of the proposed ADMM algorithm, it is sufficient to analyze the computational complexity of the per-link architecture update~\eqref{eq:prox_subproblem2_decomp}. This is because Subproblem \eqref{eq:admm_step1_sls_update} inherits the scalability of the standard localized SLS, and the dual variable update consists merely of local matrix additions.
The computational complexity of the per-link architecture update~\eqref{eq:prox_subproblem2_decomp} depends on two local metrics: the dimension of the local atomic subspace $\dim(\mathcal{E}_{ij})$ and the cardinality of the coupling set $|\mathscr{I}_{ij}|$. 

\begin{algorithm}[t]
\caption{Communication Architecture Co-Design}
\label{alg:comm-codesign}
\begin{algorithmic}[1]
    \small
    \Require System matrices $(A,B_2,C_1,D_{12})$, base topology $\Gamma_{\mathrm{base}}$, candidate edges $\mathscr{E}$, regularization weight $\lambda$, proximal ADMM parameters $(\rho, c)$, and tolerances $(\epsilon^{\mathrm{pri}}, \epsilon^{\mathrm{dual}})$.
    \Ensure Designed communication adjacency matrix $\Gamma_{\mathrm{des}}$, Optimal system response $\mathbf{\Phi}_{\mathrm{des}}^*$.
    \Statex
    \State \textbf{Initialize:} $\Gamma_{\mathrm{des}} \gets \Gamma_{\mathrm{base}}$, $\mathbf{\Phi}_{u2}^0$, $\Lambda^0$, iteration counter $k=0$.
    \Statex
    \State \textbf{Solve the regularized co-design \eqref{prob:sls_comm_codesign}:} 
    \Repeat        
        \State Update $(\mathbf{\Phi}_x^{k+1},\mathbf{\Phi}_{u1}^{k+1})$ in parallel using \eqref{eq:closed_form_sls} 
        \State Update $\mathbf{\Phi}_{u2}^{k+1}$ in parallel by computing $\mathbf{\Phi}_{u2,ij}^{k+1}$ using \eqref{eq:closed_form_arch} 
        \Statex \qquad \qquad \; and $\mathbf{\Phi}_{u2, \mathrm{base}}^{k+1} = \mathcal{P}_{\mathrm{base}}(\mathbf{\Phi}_{u1}^{k+1}+\Lambda^k)$.
        \State Update dual variable $\Lambda^{k+1}$ using \eqref{eq:admm_step3_dual_update}.
        \State $k \gets k+1$
    \Until{convergence tolerances $\epsilon^{\mathrm{pri}}, \epsilon^{\mathrm{dual}}$ are met}
    \Statex
    \State \textbf{Construct $\mathscr{E}_{\text{sel}}$ and corresponding
    communication graph:
    }
    \ForAll{$(i,j) \in \mathscr{E}$ s.t. $\|\mathbf{\Phi}_{u2,ij}^k\|_{\mathcal{H}_2} \neq 0$}
        \State $\Gamma_{\mathrm{des}} \gets \Gamma_{\mathrm{des}} \cup E_{ij}$
    \EndFor
    \Statex
    \State \textbf{Polish optimal system response on communication graph:}
    \State Solve \eqref{prob:sls_problem_with_Gamma} with $\Gamma=\Gamma_{\mathrm{des}}$ (no regularizer) to obtain $\mathbf{\Phi}_{\mathrm{des}}^*$.
    \Statex
    \State \Return $\Gamma_{\mathrm{des}},\ \mathbf{\Phi}_{\mathrm{des}}^*$
\end{algorithmic}
\end{algorithm}
Under the $d$-localized constraint $\mathcal{L}_d$, both metrics are determined by the local node density rather than the total network size.
Because $\mathcal{L}_d$ restricts the control response to a spatial radius of $d+1$, any valid shortcut path $s \to i \to j \to t$ requires $\mathrm{dist}_{\bar{A}}(s \to i) + \mathrm{dist}_{\bar{A}}(j \to t) < d$, where $\mathrm{dist}_{\bar{A}}(\cdot)$ is defined in Section~\ref{subsec:system_model}. 
This guarantees $\mathrm{dist}_{\bar{A}}(s \to i) < d$ and $\mathrm{dist}_{\bar{A}}(j \to t) < d$.
Thus, $\dim(\mathcal{E}_{ij})$ depends solely on the local node density, not the system scale.
Similarly, for any coupling link $(p,q) \in \mathscr{I}_{ij}$ to share an input-output pair $(s,t)$ with $(i,j)$, nodes $p$ and $q$ must reside within the $d$-hop neighborhoods of $i$ and $j$.
Thus, $|\mathscr{I}_{ij}|$ also depends on the local node density, independent of the global system scale.
We empirically demonstrate this scalability in Section \ref{subsec:exp_ver_scalability}.

This local coupling structure enables a fully distributed and localized implementation of the algorithm when the base communication graph $\Gamma_{\mathrm{base}}$ is undirected.
To compute the per-link architecture update \eqref{eq:prox_subproblem2_decomp} for candidate link $(i,j)$, subsystem $i$ requires two local information sets.
First, it needs the column-wise solutions $\mathbf{\Phi}_{u1}^{k+1}$ from subsystems $s$ contributing to $\mathcal{E}_{ij}$, all of which satisfy $\mathrm{dist}_{\bar{A}}(s \to i) < d$.
Second, for the coupling information, $(p,q) \in \mathscr{I}_{ij}$ implies there exists a shared path $(s,t) \in \mathcal{E}_{ij} \cap \mathcal{E}_{pq}$. Given that an undirected base topology $\Gamma_{\mathrm{base}}$ satisfies the coverage condition in Definition~\ref{def:comm_coverage_cond}, applying the triangle inequality to these shared paths yields $\mathrm{dist}_{\Gamma_{\mathrm{base}}^{\lfloor v_c \rfloor}}(p \to i) + \mathrm{dist}_{\Gamma_{\mathrm{base}}^{\lfloor v_c \rfloor}}(q \to j) < 2d$, which guarantees that $\mathrm{dist}_{\Gamma_{\mathrm{base}}^{\lfloor v_c \rfloor}}(p \to i) < 2d$.
Consequently, the per-link architecture update~\eqref{eq:prox_subproblem2_decomp} for link $(i,j)$ can be solved through local communication with neighboring subsystems within a radius of $2d-1$ on $\Gamma_{\mathrm{base}}^{\lfloor v_c \rfloor}$, ensuring that the proposed algorithm operates in a distributed and localized manner.

\section{EXPERIMENTS}
\label{sec:EXPERIMENTS}

In this section, we evaluate the performance and scalability of the proposed communication architecture co-design framework.
In all experiments, the base communication graph was set identical to the physical topology (i.e., $\Gamma_{\mathrm{base}} = \bar A$).

\subsection{Efficacy of Co-Design Algorithm}
\label{subsec:exp_performance_trade_off}
To demonstrate the efficacy of the proposed co-design algorithm, we present two sets of experiments. 
The first experiment investigates the trade-off between closed-loop performance and communication complexity, validating the effectiveness of our method against baseline search algorithms.
The second experiment visualizes the designed communication topologies and evaluates their effectiveness on a grid topology composed of MIMO subsystems.

Consider the global plant model \eqref{eq:global_plant_model}.
The elements of the state matrix $A$ are randomly generated such that the entries of the diagonal blocks are sampled from a uniform distribution over $[0.1, 1.7]$, and the entries of the off-diagonal blocks are sampled from $[-1.3, -0.1] \cup [0.1, 1.3]$.
We set $[C_1 \quad D_{12}] = I$.
In both cases, we set the spatio-temporal constraint parameters $(d,T) = (4,5)$ and the communication speed to $v_c = 1.4$.
For the proximal ADMM parameter($\rho,c$), we set $\rho = 10$ and choose $c$ appropriately to guarantee convergence, as described in Section~\ref{subsec:admm_split_and_parallel_updates} (i.e., $c \geq \rho \cdot \max_{(i,j)\in\mathscr{E}}|\mathscr{I}_{ij}|$).
The specific values of $\max_{(i,j)\in\mathscr{E}}|\mathscr{I}_{ij}|$ for both experimental cases can be found in Section~\ref{subsec:exp_ver_scalability}.

\subsubsection{Trade-off Analysis on a Ring Topology}

\begin{figure}[t]
    \centering
    \includegraphics[width=\linewidth]{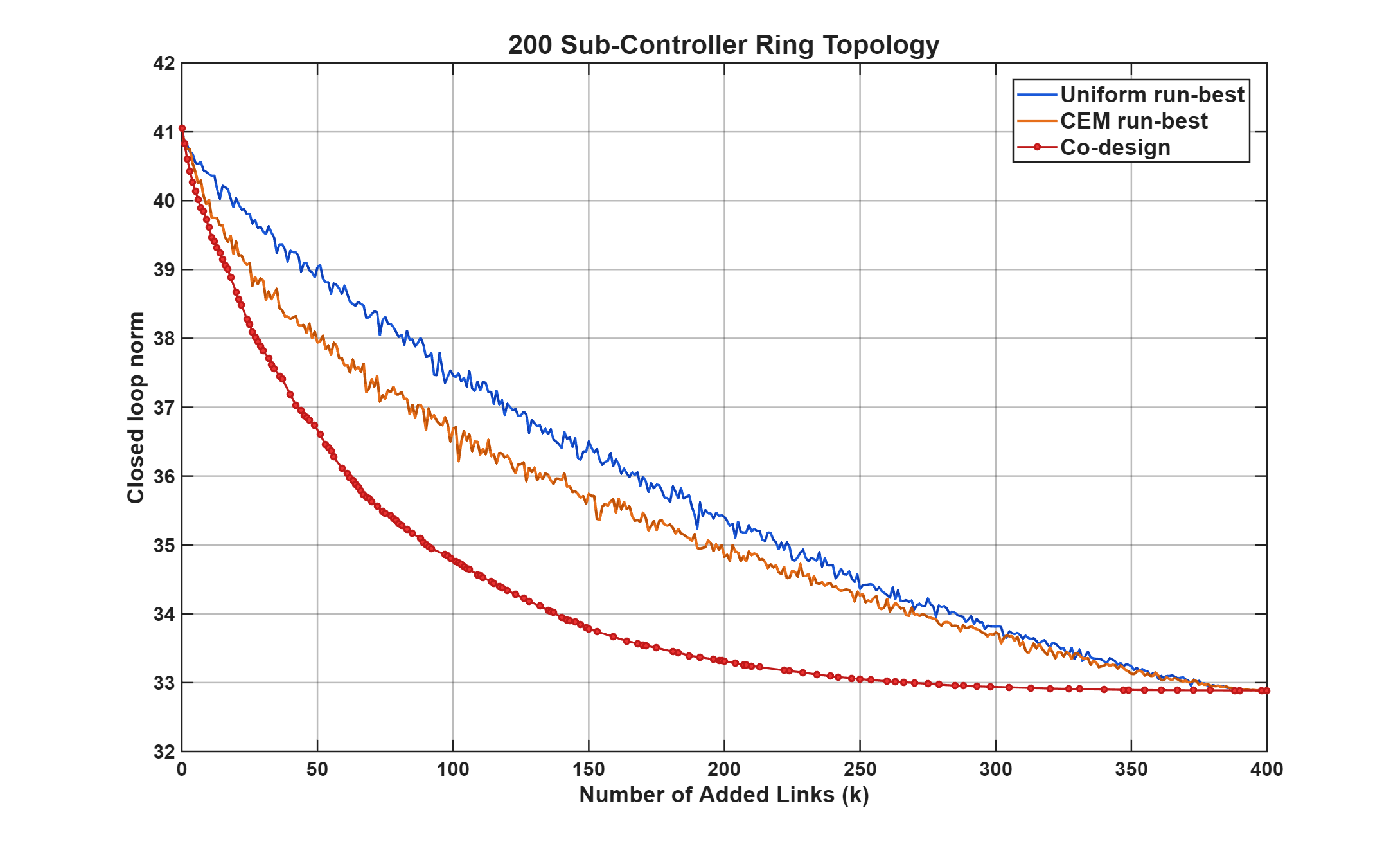}
    \caption{
    The red line represents the closed-loop $\mathcal{H}_2$ performance of the distributed optimal controller on communication graphs designed via Algorithm~\ref{alg:comm-codesign}. The orange and blue lines indicate the best performance found by the CEM and uniform search algorithms, respectively, for each fixed number of additional links.}
    \label{fig:chain_opt_compare}
\end{figure}

We conducted experiments on a 200-node ring topology. 
We assumed each subsystem $i$ had a scalar state $x_i \in \mathbb{R}$, setting $B_2=I_{N_s}$.
The spectral radius of $A$ was $3.2053$. The candidate edge set $\mathscr{E}$ comprised 2-hop neighbor connections. 
We executed the co-design algorithm detailed in Algorithm~\ref{alg:comm-codesign} by sweeping the regularization weight $\lambda$ over the interval $[0, 30]$, while setting the proximal ADMM parameter $c$ to 200.

\begin{figure}[t]
    \centering
    \includegraphics[width=0.7\linewidth]{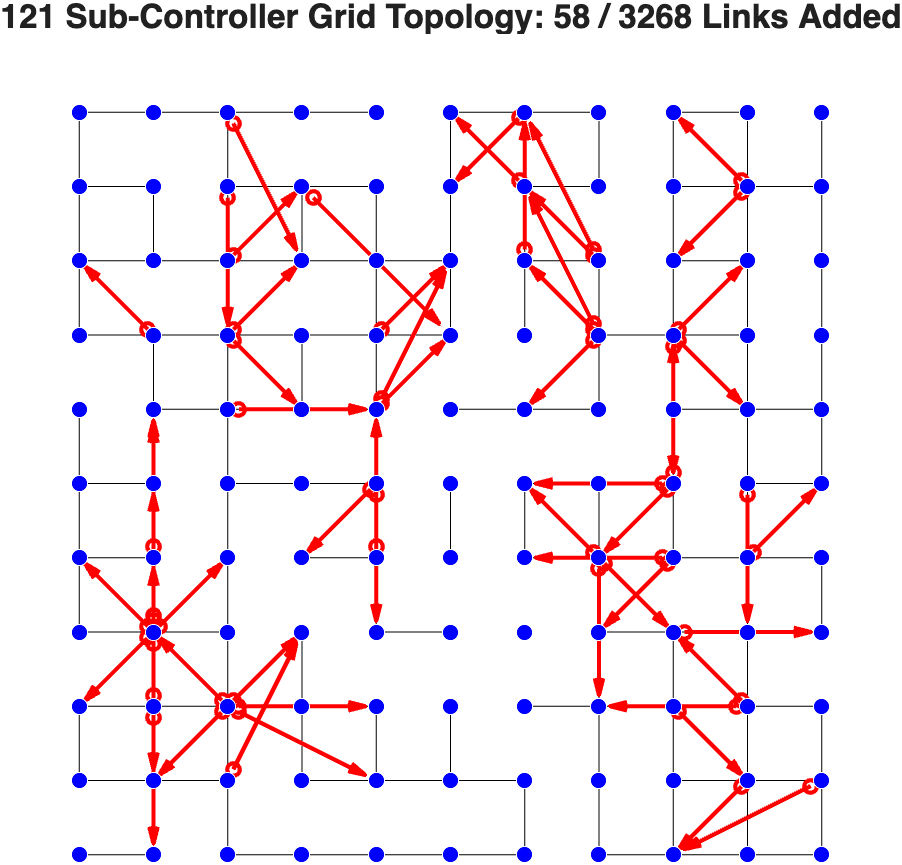}
    \caption{Blue dots represent the subsystems, and black lines indicate the physical topology, which serves as the base communication topology. Red arrows, with circled origins, denote the additional communication edges selected by the co-design algorithm, activating only 58 out of 3,268 candidate edges.}
    \label{fig:grid_topology_mimo}
\end{figure}

To establish baselines, we employed the Cross-Entropy Method (CEM) \cite{de2005tutorial}---a prominent stochastic search algorithm for combinatorial optimization---and a uniform distribution search algorithm to find the optimal sparse communication links. To ensure a rigorous comparison, we established the computational complexity of our ADMM-based solver in terms of equivalent full SLS evaluations and utilized this as the lower bound (LB) for the search algorithms' budget. Specifically, the total search budget for the baselines was set to $500,000$ evaluations, securely exceeding this $LB= 177,600$. 

Because in the proximal ADMM, both the column-wise SLS subproblems and the proximal subproblems admit closed-form updates, we approximate their individual computational costs to be equivalent. Therefore, the computational cost of a single ADMM iteration can be normalized by the cost of solving one full SLS problem (which consists of $N_{\mathrm{SLS}}$ subproblems). The LB of the total equivalent SLS evaluations was formulated as:
\begin{equation*}
    LB = \left( \overline{I}_{\mathrm{ADMM}} \times \left( 1+ \frac{N_{\mathrm{Arch}}}{N_{\mathrm{SLS}}}\right) \right) \times N_{\mathrm{edge}},
\end{equation*}
where $\overline{I}_{\mathrm{ADMM}}$ is the average number of ADMM iterations during $\lambda$ sweeping, $N_{\mathrm{SLS}}$ is the number of column-wise SLS subproblems, $N_{\mathrm{Arch}}$ is the number of the per-link architecture update~\eqref{eq:prox_subproblem2_decomp}, and $N_{\mathrm{edge}}$ is the number of candidate edges.
In this experiment, $\overline{I}_{\mathrm{ADMM}}= 147.99$, and $N_{\mathrm{edge}}=400$, therefore $LB= 177,600$.
Finally, the total search budget was distributed among the specific number of additional links proportionally to the logarithm of the combinatorial size of selecting those links from the candidate edge set.
Furthermore, we ensured that all evaluated samples were unique, avoiding any redundant computations within the given budget.

As shown in Fig. \ref{fig:chain_opt_compare}, increasing the regularization weight $\lambda$ results in simpler communication architectures through the co-design procedure.
Furthermore, for any fixed number of additional links, the distributed optimal controller implemented on the communication topology designed by the proposed co-design framework consistently outperforms the controller operating on the best topology obtained from the search algorithms.

\subsubsection{Example of Co-Design on a Grid Topology}
We evaluated the framework on a $11 \times 11$ grid topology whose subsystems have two-dimensional states. To construct the physical topology $\bar{A}$, edges between adjacent nodes were randomly generated with a probability of $0.7$, and the spectral radius of the state matrix $A$ was $3.6310$.
We assumed fully state-actuated subsystems ($B_2 = I_{N_s} \otimes I_2$).
For the candidate edge set $\mathscr{E}$, we considered all directed links between nodes with a hop distance of 1 to 4 that were not already present in $\Gamma_{\mathrm{base}}$, resulting in a total of $3,268$ candidate edges.
We set the regularization weight $\lambda = 100$ and the proximal ADMM parameter $c = 3,500$.

The co-design algorithm activated 58 of the 3,268 candidate edges (1.77\%), as shown in Fig.~\ref{fig:grid_topology_mimo}. 
Despite this highly sparse link addition, the co-designed topology yielded a closed-loop $\mathcal{H}_2$ norm of 121.4782, achieving a substantial 18.83\% reduction compared to the base topology's norm of 149.6607.

\subsection{Verification of Scalability} \label{subsec:exp_ver_scalability}
In the Sections~\ref{subsec:dist_impl}, we discussed that the computational complexity of the per-link architecture update~~\eqref{eq:prox_subproblem2_decomp} is independent of the global system scale, which guarantees the scalability of the proposed co-design framework.
To empirically verify this scalability, we evaluated the maximum dimension of the atomic subspace, $\max_{(i,j)} \dim(\mathcal{E}_{ij})$, and the maximum cardinality of the coupling set, $\max_{(i,j)} |\mathscr{I}_{ij}|$, across different network topologies. 

We conducted this analysis on both ring and grid topologies. For both cases, we fixed the spatio-temporal parameters as $(d,T)=(4,5)$, $v_c=1.4$ and incrementally increased the total number of subsystems $N_s$. 
For the ring topology, the candidate edge set $\mathscr{E}$ was generated using the 2-hop neighborhood. For the grid topology experiments, we considered two distinct settings for $\mathscr{E}$: the 8-neighbor (Moore) neighborhood, and all directed links with a hop distance of 1 to 4 excluding those in $\Gamma_{\mathrm{base}}$.
Each configuration was evaluated across 3 independent random seeds.

\begin{table}[ht]
\caption{ The maximum atomic subspace dimension and the maximum coupling set cardinality evaluated across different system scales for ring and grid topologies.}
\label{tab:scalability_metrics}
\centering
\begin{tabular}{c c c c}
\toprule
\textbf{Topology} & \makecell{\textbf{System scale} \\ ($N_s$)} & $\max \dim(\mathcal{E}_{ij})$ & $\max |\mathscr{I}_{ij}|$ \\
\midrule
\multirow{3}{*}{Ring} 
 & 50  & 7 & 6 \\
 & 100 & 7 & 6 \\
 & 200 & 7 & 6 \\
\midrule
\multirow{3}{*}{\makecell{Grid \\ (8-neighbor)}}
 & 121 ($11 \times 11$) & 34.0($\pm$ 7.8) & 34.6($\pm$ 3.0) \\
 & 196 ($14 \times 14$) & 36.6($\pm$ 2.5) & 38.0($\pm$ 2.6) \\
 & 324 ($18 \times 18$) & 35.0($\pm$ 1.0) & 39.3($\pm$ 2.0) \\
\midrule
\multirow{3}{*}{\makecell{Grid \\ (Hop 1-4)}} 
 & 121 ($11 \times 11$) & 41.3($\pm$ 2.3) & 293.0($\pm$ 14.0) \\
 & 196 ($14 \times 14$) & 45.6($\pm$ 3.0) & 320.6($\pm$ 16.5) \\
 & 324 ($18 \times 18$) & 44.3($\pm$ 2.8) & 310.6($\pm$ 22.5) \\
\bottomrule
\end{tabular}
\end{table}

The results are summarized in Table \ref{tab:scalability_metrics}.
As the number of subsystems $N_s$ scales up from $50$ to $200$ in the ring topology, both $\max \dim(\mathcal{E}_{ij})$ and $\max |\mathscr{I}_{ij}|$ remain completely constant.
Likewise, for the grid topology, these metrics remain nearly constant across varying system scales ($N_s = 121$ to $324$) within each candidate edge configuration, but differ distinctly between the two configurations.
Aligning with our discussion in Section~\ref{subsec:dist_impl}, the empirical data demonstrates that $\max \dim(\mathcal{E}_{ij})$ and $\max |\mathscr{I}_{ij}|$ depend on the local node density and are not affected by the global system scale.

\section{CONCLUSION} \label{sec:conclusion}

We proposed a scalable and convex framework for the co-design of communication architectures and distributed controllers within the SLS paradigm.
To this end, we characterized a communication topology-induced feasible subspace of system responses via adjacency matrix $\Gamma$ and a communication speed $v_c$. 
Based on this characterized subspace, we proposed the SLS formulation for communication graph co-design.
By leveraging the proximal ADMM, this problem can be solved in a scalable, distributed and localized manner.
Through numerical experiments, we demonstrated that this framework effectively navigates the trade-off between closed-loop performance and communication architecture complexity, with a computational cost that remains independent of the global system scale.
Future research directions include relaxing the ideal-link assumption to heterogeneous link speeds and to communication uncertainties, and extending the framework to measured output-feedback settings with partial state measurements and measurement noise.
Finally, it would be of interest to investigate the selection of a highly sparse, non-physical base graph $\Gamma_{base}$ and the subsequent link-addition methods necessary to ensure that the ultimate outcome $\Gamma_{des}$ satisfies the coverage condition.




\addtolength{\textheight}{-12cm}



\bibliographystyle{IEEEtran}           
\bibliography{IEEEabrv, refs}        

\end{document}